%% file: paper3.tex
\include{settings}	

\ifcompileiliad
	\documentclass[10pt]{article}
\else
	\documentclass[12pt]{article}
\fi

\ifcompileiliad
	\usepackage[papersize={122mm,160mm},
				hmargin=0.2cm,
				vmargin=0.2cm,
				tmargin=1.0cm,
				bmargin=0.4cm,
				dvips]
				{geometry}
\else
	\usepackage[paper = a4paper,         
				twoside,                 
				bindingoffset = 2mm,     
				hmargin = 25mm,          
				vmargin = 25mm,          
				dvipdfm]
				{geometry}
\fi

\usepackage[scaled]{helvet} 
\usepackage{sectsty}        
\allsectionsfont{\usefont{OT1}{phv}{bc}{n}\selectfont} 
\usepackage{bm}
\usepackage{graphicx}       
\usepackage[utf8]{inputenc} 
\usepackage{amsmath}
\usepackage{amsfonts}
\usepackage{amsthm}
\usepackage[mathcal]{euscript} 
\usepackage{booktabs}       
\usepackage{setspace}       
\usepackage{forloop}		
\usepackage[round]{natbib}       
\usepackage{subcaption}

\usepackage[procnames]{listings} 

\include{macros}

\newtheorem{thm}{Theorem}

\theoremstyle{definition}
\newtheorem{ass}{Assumption}[section]

\renewcommand{\vec}[1]{\bm{#1}}

\begin{document}

\title{Non-parametric estimation of conditional densities: A new method}

\author{H\r{a}kon Otneim \footnote{\noindent Norwegian School of Economics, Department of Business and Management Science, Helleveien 30, 5045 Bergen, Norway. \newline  E-mail: \texttt{hakon.otneim@nhh.no}} \and Dag Tj{\o}stheim \footnote{University of Bergen, Department of Mathematics, P.B. 7803, 5020 Bergen, Norway}}
\date{\vspace{-5ex}}
\maketitle

\begin{abstract}
Let $\vec{X} = (X_1,\ldots,X_p)$ be a stochastic vector having joint density function $f_{\vec{X}}(\vec{x})$ with partitions $\vec{X}_1 = (X_1,\ldots,X_k)$ and $\vec{X}_2 = (X_{k+1},\ldots,X_p)$. A new method for estimating the conditional density function of $\vec{X}_1$ given $\vec{X}_2$ is presented. It is based on locally Gaussian approximations, but simplified in order to tackle the curse of dimensionality in multivariate applications, where both response and explanatory variables can be vectors. We compare our method to some available competitors, and the error of approximation is shown to be small in a series of examples using real and simulated data, and the estimator is shown to be particularly robust against noise caused by independent variables. We also present examples of practical applications of our conditional density estimator in the analysis of time series. Typical values for $k$ in our examples are 1 and 2, and we include simulation experiments with values of $p$ up to 6. Large sample theory is established under a strong mixing condition.
\end{abstract}


\section{Introduction}
 \label{sec:introduction}
The need for expressing statistical inference in terms of conditional quantities is ubiquitous in most natural and social sciences. The obvious example is the estimation of the mean of some set of response variables conditioned on sets of explanatory variables taking specified values. Other common tasks are the forecasting of volatilities or quantiles of financial time series conditioned on past history. Problems of this kind often call for some sort of regression analysis, of which the literature provides an abundance of choices. 

Conditional means, variances and quantiles are all properties of the conditional density, if it exists, as are all other probabilistic statements that we might ever want to make about the response variables given the explanatory variables. It is therefore clearly of interest to obtain good estimates of the entire conditional distribution in order to make use of all the evidence contained in the data, and to provide the user with a wide variety of options in analysing and visualising the relationships of the variables under study.

The classical method for non-parametric density estimation is the kernel estimator \citep{rosenblatt1956remarks, parzen1962estimation}, which in the decades following its introduction has been refined and developed in many directions. Especially the crucial choice of smoothing parameter, or bandwidth, has been addressed by several authors, including \cite{silverman1986density}, \cite{sheather1991reliable} and \cite{chacon2010multivariate}. The kernel estimator suffers greatly from the curse of dimensionality however, which quickly inhibits its use in multivariate problems. Several alternative methods of estimation has been proposed to improve performance if the subject of estimation is a joint multivariate density function, most recently the LGDE (locally Gaussian density estimator) by \cite{otneim2016locally}, which the work in the present paper takes as its starting point. Very few methods exist for the non-parametric estimation of conditional densities though, especially if we do not wish to restrict ourselves to cases with one-dimensional response and/or explanatory variables. This lack of methodology is surprising, considering the aforementioned importance of estimating conditional densities; the practical use of which is of altogether greater interest than unconditional density estimates, as is illustrated by some of its possible applications in Section \ref{sec:examples}.   

In this paper we present a new method for estimating conditional densities based on local Gaussian approximations. Let $\vec{X} = (X_1, \ldots, X_p)$ be a stochastic vector, and, assuming existence, denote by $f_{\vec{X}}(\cdot)$ its joint density function. Further, let $(\vec{X}_1; \vec{X}_2) = (X_1, \allowbreak \ldots, X_k; X_{k+1}, \ldots, X_p)$ be a partitioning of $\vec{X}$. Then the conditional density of $\vec{X}_1$ given $\vec{X}_2 = \vec{x}_2$ is defined by
\begin{equation}\label{eq:cond}f_{\vec{X}_1|\vec{X}_2}(\vec{x}_1|\vec{X}_2 = \vec{x}_2) = \frac{f_{\vec{X}}(\vec{x}_1,\vec{x}_2)}{f_{\vec{X}_2}(\vec{x}_2)},\end{equation}
where $f_{\vec{X}_2}$ is the marginal density of $\vec{X}_2$.

The problem of estimating (\ref{eq:cond}) is not trivial. We do not observe data directly from the density that we wish to estimate, so we need a different set of tools than those used in the unconditional case. A natural course of action is to follow \cite{rosenblatt1969conditional} in obtaining good estimates of the numerator and denominator of (\ref{eq:cond}) separately using the kernel estimator, and use the definition directly. \citet{chen2001estimation} provide a discussion of choosing the bandwidths when using the kernel estimator to estimate the components, as do \cite{bashtannyk2001bandwidth}. \citet[chap. 5]{li2007nonparametric} give a unified approach to estimating conditional densities using the kernel estimator, which allows a mix of continuous and discrete variables, and automatically smooths out the irrelevant ones. 

Unless one has a very good estimate of the marginal density, however, it is less than ideal to put a kernel estimate in the denominator of (\ref{eq:cond}). This is remedied by \cite{faugeras2009quantile}, who writes the conditional density as a product of the marginal and copula density functions in the bivariate case,
\begin{equation}\label{eq:condcop}
f_{X_1|X_2}(x_1|X_2=x_2) = f_{X_1}(x_1)c\left\{F_{1}(x_1), F_{2}(x_2)\right\},
\end{equation} 
where $f_{X_1}$ is the marginal density of $X_1$, $F_{1}$ and $F_{2}$ are the marginal distribution functions, $c$ is the copula density of $(X_1, X_2)$, and estimates those separately using the kernel estimator. The formula (\ref{eq:condcop}) can be generalized to the case of several covariates, but its practical use in higher dimensions is questionable because of boundary and dimensionality issues, unless one obtains better estimates of the multivariate copula density than provided by the kernel estimator, such as the local likelihood approach by \cite{geenens2014probit}.  

\citet{hyndman1996estimating} starts to move away from the kernel estimator by adjusting the conditional mean to match a better performing regression technique, such as local polynomials, while \citet{fan1996estimation} estimate the conditional density directly using locally linear and locally quadratic fits, a method that \citet{hyndman2002nonparametric} refine by constraining it to always be non-negative. The latter authors propose in the same paper a local likelihood approach which is based on some of the same machinery as we will employ in this paper, and \citet{fan2004crossvalidation} provide a cross-validation rule for bandwidth selection in the locally parametric models. These methods are to date implemented in the bivariate case only, however, where the response- and explanatory variables are both scalars.

Indeed, the main motivation behind our new method is to provide an estimator that can handle a greater number of variables without the requirement that either response or explanatory variables are scalar.

\citet{holmes2012fast} develop a fast bandwidth selection algorithm, while correctly pointing out that bandwidth selection is a formidable computational and time-consuming task in non-parametric multivariate density estimation. We argue that the curse of dimensionality is an even bigger problem, because it will not be solved by clever algorithms, but is an inherent problem in all non-parametric analysis. We therefore base our method on the newly developed locally Gaussian density estimator (LGDE) \citep{otneim2016locally}, which shows a promising robustness against dimensionality issues when estimating the multivariate unconditional density function. By exploiting locally the property of the Gaussian distribution that conditional densities are again Gaussian, we will see that conditional density estimates are readily available from the LGDE.    

This paper is organized as follows: In Section \ref{sec:LGDE} we give a short introduction to the LGDE method for multivariate \emph{un}conditional density estimation, and in Section \ref{sec:estimation} we show that extracting conditional density estimates from the LGDE is straightforward and requires neither additional estimation steps, nor integration over the joint density estimate. In Section \ref{sec:asymp} we derive some large-sample properties for our estimator under a strong mixing condition, and proceed in Section \ref{sec:examples} with a series of examples using real and simulated data, indicating the wide potential of conditional density estimation. Some concluding remarks and suggestions for further research follow in Section \ref{sec:conclusion}, and we include an appendix that contains the technical proofs.   

\section{A brief introduction to the LGDE} \label{sec:LGDE}
Because of its close relationship with our conditional density estimator, we include here a basic account of the LGDE. Suppose that we wish to estimate the full $p$-variate density $f_{\vec{X}}$ based on $n$ independent observations $\vec{X}_1, \ldots,\vec{X}_n$. \citet{hjort1996locally} provide a general setup for fitting a parametric family of densities $\psi(\cdot, \vec{\theta})$ \emph{locally} to the unknown density by maximising the local log-likelihood function in each point $\vec{x}$;
\begin{equation}
\widehat{\vec{\theta}}(\vec{x}) = {\arg\max}_{\vec{\theta}}\,\,n^{-1}\sum_{i=1}^nK_{\vec{h}}(\vec{X}_i- \vec{x})\log\psi(\vec{X}_i, \vec{\theta}) - \int K_{\vec{h}}(\vec{y} - \vec{x})\psi(\vec{y}, \vec{\theta})\,\textrm{d}\vec{y},\label{eq:loclik}
\end{equation}
so that the estimated density is given by $\widehat f_{\vec{X}}(\vec{x}) = \psi(\vec{x},\allowbreak \widehat{\vec{\theta}}(\vec{x}))$. We use standard notation, letting $\vec{h}$ denote a diagonal matrix of bandwidths, $K(\cdot)$ a symmetric kernel function integrating to one, and $K_{\vec{h}}(\vec{x}) = |\vec{h}|^{-1}K(\vec{h}^{-1}\vec{x})$. 
Denote by $\phi$ and $\Phi$ the univariate standard normal density and distribution functions respectively,
$$\phi(z) = (2\pi)^{-1/2}\exp\left\{-z^2/2\right\}, \,\, \Phi(z) = \int_{-\infty}^{z}\phi(y)\,\textrm{d}y.$$
According to \cite{otneim2016locally}, we can write the $p$-variate density function $f_{\vec{X}}$ as
\begin{equation}
f_{\vec{X}}(\vec{x}) = f_{\vec{Z}}\left(\Phi^{-1}\left(F_1(x_1)\right), \ldots,\Phi^{-1}\left( F_p(x_p) \right)\right)\prod_{i=1}^p \frac{f_i(x_i)}{\phi\left(\Phi^{-1}\left( F_i(x_i)\right)\right)} \label{eq:transformeddensity} 
\end{equation} 
where $f_i$ and $F_i$, $i = 1,\ldots,p$, are the marginal densities and distribution functions of $f_{\vec{X}}$, and  $f_{\vec{Z}}$ is the density function of a stochastic vector $\vec{Z} = (Z_1,\ldots, Z_p)$ with standard normal margins, and $Z_i = \Phi^{-1}\left(F_i(X_i)\right)$.  

We estimate $f_{\vec{Z}}$ by locally fitting the standardized normal distribution,  
\begin{equation}
\psi(\vec{z},\vec{\theta}) = \psi(\vec{z},\vec{R}) = (2\pi)^{-p/2}|\vec{R}|^{-1/2}\exp\left\{-\frac{1}{2}\vec{z}^T\vec{R}^{-1}\vec{z}\right\},\label{eq:gaussfam}
\end{equation}
with $\vec{R} = \vec{R}(\vec{z}) = \{\rho_{ij}(\vec{z})\}$ denoting the correlation matrix, based on the marginally Gaussian pseudo-observations 
\begin{equation}
\widehat{\vec{Z}}_j = \left(\Phi^{-1}(\widehat F_1(X_{j1})), \ldots, \Phi^{-1}(\widehat F_p(X_{jp}))\right)^T,\qquad j = 1,\ldots, n,\label{eq:pseudo}
\end{equation}
where  $\widehat{F}_k(x_k), k = 1,\ldots,p$ are estimates of the marginal distribution functions, which, in our asymptotic results are assumed to be the empirical marginal distribution functions. There are several reasons for transforming all observation vectors to the standard Gaussian scale. First of all, it makes the choice of the Gaussian distribution as local parametric family in (\ref{eq:gaussfam}) very natural, where, in particular, we have fixed all means and standard deviations so they are equal to $0$ and $1$ respectively, reducing the number of parameters that we must estimate locally. Moreover, the normalisation (\ref{eq:pseudo}) is a quick way to make the data more tidy, because the pseudo-observations will all be on the same scale, and there will not be any outliers which is otherwise known to create problems when using cross-validation to select bandwidths \citep{hall1987kullback}. In general, distributions become easier to estimate when they are closer to the Gaussian distribution, as shown and exploited by \cite{wand1991transformations} and \cite{ruppert1994bias}.

In (\ref{eq:gaussfam}), each correlation $\rho_{ij}(\vec{z})$ depends on the coordinates of the entire $\vec{z}$-vector, making its estimation difficult because of the curse of dimensionality. In regression problems, this issue may be tackled by imposing an additive structure on the unknown regression function:
$$Y=f(X_1,\ldots,X_p) + \epsilon = f_1(X_1) + \cdots + f_p(X_p) + \epsilon,$$
possibly with higher order interactions if the data can support it. One motivation behind the LGDE is to introduce a similar idea to density estimation, and it is based on the fact that a \emph{global} Gaussian fit is produced by calculating the correlation coefficients between each pair of variables by using only the corresponding observation vectors.  In order to circumvent the curse of dimensionality, \cite{otneim2016locally} carry this procedure over to the local case by restricting $\rho_{ij}(\vec{z})$ so that it
is only allowed to depend on its own variables; i.e. $\rho_{ij}(\vec{z}) = \rho_{ij}(z_i, z_j)$. The corresponding estimate $\widehat{\rho}(z_i, z_j)$ is computed from the corresponding simplified pairwise local log likelihood so that we can take
\begin{equation} \label{eq:simplification}
\widehat\rho_{ij}(z_1,\ldots,z_p) = \widehat\rho_{ij}(z_i,z_j).
\end{equation} 
This technique effectively reduces the estimation of $f_{\vec{X}}$ to a series of bivariate local problems, which is reflected in the rate of convergence in the following asymptotic result, that holds under some standard regularity conditions \citep{otneim2016locally} and proven for sets of iid observations:
\begin{equation} \label{eq:anormLGDE}
\sqrt{nh_n^2}\left(\widehat f_{\vec{X}}(\vec{x}) - f_0(\vec{x})\right) \stackrel{\mathcal{L}}{\rightarrow}N\left(0, \sigma^2_{f_{\vec{X}}}\right),
\end{equation}
 where, in general, $f_0(\vec{x})\neq f_{\vec{X}}(\vec{x})$ is the population density towards which the LGDE converges. Here, $f_0(\vec{x})$ is the simplified density obtained from (\ref{eq:transformeddensity}) and (\ref{eq:gaussfam}) by replacing $f_{\vec{Z}}(\vec{z})$ with $\Psi(\vec{z},\vec{R}_0)$, where $\vec{R}_0=\{\rho_{0,ij}(z_i,z_j)\}$ and $\rho_{0,ij}$ is the true local Gaussian correlation between $Z_i$ and $Z_j$, as will be defined in Section \ref{sec:asymp}. 

\cite{otneim2016locally} propose two methods for bandwidth selection. Cross-validation is used to determine the bandwidths that minimise the estimated Kullback-Leibler distance between the density estimate and the true density. They also employ the $k$-nearest neighbor technique in order to obtain adaptive bandwidths, but simulation results suggest that, of the two, the global bandwidth selector performs better. Indeed, as already mentioned, \cite{hall1987kullback} shows that the performance of cross-validation bandwidth selection depends on the tails of the underlying distribution not being thicker than the tails of the kernel function. By transforming the data to marginal standard normality, and using the Gaussian kernel function, it follows that the cross-validation procedure is well suited for selecting the LGDE bandwidths. 

\section{Estimating the conditional density} \label{sec:estimation}
Conditional density estimates are in principle available from any non-parametric estimate of the unconditional density of all variables. Let us return to the problem in Section \ref{sec:introduction}, and suppose that we obtain an estimate $\tilde{f}_{\vec{X}}$ of $f_{\vec{X}}$ in the process of estimating the left hand side of (\ref{eq:cond}). The corresponding marginal density $\tilde{f}_{\vec{X}_2}$ that ideally we should put in the denominator of (\ref{eq:cond}) is given by
$$\tilde{f}_{\vec{X}_2} = \int \tilde f_{\vec{X}} \,\textrm{d}\vec{x}_1,$$
but one must usually turn to numerical methods in order to obtain this integral, which can be a costly affair in terms of computing power, especially when there are many variables over which to integrate. Thus, estimating the marginal density directly from the data is often quicker, but introduces a new source of uncertainty that, again, will be difficult to handle in case of several explanatory variables.

We proceed to show that this problem is completely circumvented if we use the LGDE strategy for estimation. As is well known for a multivariate Gaussian distribution, every conditional density that can be formed by partitioning the Gaussian vector and computing the fraction (\ref{eq:cond}), is again Gaussian, and where the (conditional) mean and (conditional) covariance matrix in that Gaussian can be easily computed; see e.g. \citet[Chap. 4]{johnson2007applied}. This is of course also the case for the fraction of Gaussians that are local approximations, and we can obtain estimates by using these formulas. In more detail, starting from the $p$-variate density in (\ref{eq:transformeddensity}),

\begin{align*}f_{\vec{X}_1|\vec{X}_2}(\vec{x}_1|\vec{X}_2 = \vec{x}_2) &= \frac{ f_{\vec{X}}(\vec{x})}{ f_{\vec{X}_2}(\vec{x}_2)} \\
&= \frac{f_{\vec{Z}}(z_1,\ldots,z_p)}{f_{\vec{Z}_2}(z_{k+1},\ldots,z_p)} \prod_{i=1}^k \frac{ f_i(x_i)}{\phi\left(z_i\right)},\end{align*}
where $f_{\vec{Z}}/f_{\vec{Z}_2}$ can be seen locally as a fraction of a $p$-variate and a $p-k$-variate Gaussian function, each with all expectations equal to zero, and with correlation matrices $\vec{R}(\vec{z})$ and $\vec{R}_{22}(\vec{z})$ respectively. The latter notation is natural because of the pairwise analysis, so that $\vec{R}_{22}(\vec{z})$ is \emph{exactly equal} to the lower right block of $\vec{R}(\vec{z})$. Thus, in every grid point $\vec{z}$, $f_{\vec{Z}_2}$ is exactly the marginal density of the $p-k$ last variables of $f_{\vec{Z}}$, and we can use the basic result for the multivariate normal distribution mentioned above to rewrite the fraction. Partition $\vec{R}(\vec{z})$ into four blocks, of which the lower right block is $\vec{R}_{22}(\vec{z})$:
$$\vec{R}(\vec{z}) = \begin{pmatrix}\vec{R}_{11} & \vec{R}_{12}\\ \vec{R}_{21} & \vec{R}_{22} \end{pmatrix}$$
Then 
\begin{equation} \label{eq:fraction}
f_{\vec{Z}}/f_{\vec{Z}_2} = \Psi^*(z_1, \ldots, z_k; \vec{\mu}^*, \vec{\Sigma}^*),
\end{equation}
where $\Psi^*(\cdot)$ is the general $k$-variate Gaussian density with expectation vector and covariance matrix given by 
\begin{align}
\vec{\mu}^* &= \vec{R}_{12}\vec{R}_{22}^{-1}\vec{z}_2, \label{eq:locmean}\\
\vec{\Sigma}^* &= \vec{R}_{11} - \vec{R}_{12}\vec{R}_{22}^{-1}\vec{R}_{21} \label{eq:locvar},
\end{align}
where $\vec{z_2} = (z_{k+1}, \ldots, z_{p})$. Note that we may use correlation- and covariance matrices interchangeably, because all standard deviations are equal to one in $f_{\vec{Z}}$ and $f_{\vec{Z}_2}$. 

We can now obtain an estimate of $f_{\vec{X}_1|\vec{X}_2=\vec{x}_2}$ by plugging in local likelihood estimates of $\vec{R}(\vec{z}) = \{\rho_{ij}(z_i, z_j)\}$, resulting in 
\begin{equation}\widehat f_{\vec{X}_1|\vec{X}_2}(\vec{x}_1|\vec{X}_2 = \vec{x}_2)= \Psi^*\left(\widehat{\vec{z}}; \widehat{\vec{\mu}^*}(\widehat{\vec{z}}),\widehat{\vec{\Sigma}^*}(\widehat{\vec{z}}) \right) \prod_{i=1}^k \frac{\widehat {f}_i(x_i)}{\phi\left(\widehat{z_i}\right)},\label{eq:backtrans}\end{equation}
where $\widehat{\vec{\mu}^*}(\widehat{\vec{z}})$ and $\widehat{\vec{\Sigma}^*}(\widehat{\vec{z}})$ are obtained by substituting local correlation estimates into equations (\ref{eq:locmean}) and (\ref{eq:locvar}), and where we write $\widehat{z}_i = \Phi^{-1}(\widehat F_i(x_i))$. Moreover, the second factor in (\ref{eq:backtrans}) requires estimates $\widehat f_i(x_i)$ of the marginal densities $f_i(x_i)$, $i=1,\ldots,k$. As we will see in the next section, this can be any smooth estimate, and will not affect the asymptotic results as long as they converge faster than $\sqrt{nh^2}$. The current implementation of the LGDE uses the logspline estimator by \cite{stone1997polynomial} for this purpose. It is interesting to note that the computation resulting in (\ref{eq:fraction}), (\ref{eq:locmean}) and (\ref{eq:locvar}) can be done directly on estimated quantities using results on fractions of exponential functions.  

We modify the LGDE algorithm in \cite{otneim2016locally} according to the discussion above, and estimate conditional densities by following these steps:
\begin{itemize}
\item[1.] Transform each marginal observation vector to pseudo-standard normality using (\ref{eq:pseudo}).
\item[2.] Estimate the local correlation matrix of the transformed data by fitting the Gaussian family (\ref{eq:gaussfam}) using the local likelihood function in (\ref{eq:loclik}) and the simplification (\ref{eq:simplification}). In practice, this amounts to fitting the bivariate version of (\ref{eq:gaussfam}) to each pair of approximately marginally standard normal variables $(\widehat{Z}_i, \widehat{Z}_j)$, and let $\widehat{\vec{R}}(\vec{z}) = \{\widehat{\rho}(z_i, z_j)\}_{i,j = 1,\ldots,p}$.
\item[3.] Calculate the local mean and covariance matrix of $\widehat f_{\vec{Z}}/\widehat f_{\vec{Z}_2}$ using the formulas (\ref{eq:locmean}) and (\ref{eq:locvar}), so that the conditional density estimate becomes as given in (\ref{eq:backtrans})

\item[4.] Normalize the density estimate so that it integrates to one.
\end{itemize}
Again, we point out that our simplification of the dependence structure (\ref{eq:simplification}) in general will result in an estimate of an approximation $f_0(\cdot)$ of the true density $f(\cdot)$. We proceed in the next section to discuss the nature of the simplification, to discuss regularity conditions, and to explore the large sample properties of our method. 

\section{Regularity conditions and asymptotic theory}
\label{sec:asymp}
The following theorems on consistency relative to $f_0$ and asymptotic normality state analogous results to those found in \cite{otneim2016locally}, but they are proven under a new set of regularity conditions that allow for dependence between the observations $X_1, \ldots, X_n$. 

The simplification (\ref{eq:simplification}) means that we estimate the local correlations pairwise, which also means that it suffices to derive most of the asymptotic theory in the bivariate case. Consider, for the time being, a pair $(Z_i, Z_j)$ of marginally standard normal random variables. Denote by $\rho_0(z_i,z_j)=\rho_0(\vec{z})$ the local Gaussian correlation between them, as will be defined below, and by $\widehat\rho(\vec{z})$ its estimate, calculated using the bandwidths $\vec{h}=(h_i,h_j)$ according to the algorithm in Section \ref{sec:estimation}. Denote further by $L_n(\rho(\vec{z}),\vec{z})$ the local log-like\-li\-hood function in (\ref{eq:loclik}) with the bivariate version of (\ref{eq:gaussfam}) as parametric family $\psi(\cdot, \rho)$. For a fixed $\vec{h}>0$ (where all statements about the vector $\vec{h}$ in this section are element-wise), denote by $\rho_{\vec{h}}$ the local correlation that satisfies
\begin{equation}
\frac{\partial L_n(\vec{\rho}; \vec{z})}{\partial \rho} \rightarrow \int K_h(\vec{y} - \vec{z})u(\vec{y}, \rho_{\vec{h}})\left\{ f_{ij}(\vec{y}) - \psi(\vec{y}, \rho_{\vec{h}})\right\}\,\textrm{d}\vec{y}  = 0 \label{eq:limit}
\end{equation} 
as $n\rightarrow\infty$, where $u(\cdot,\rho) = \partial\log\psi(\cdot,\rho)/\partial\rho$, and $f_{ij}$ is the joint density of $(Z_i,Z_j)$. We assume hereafter that $\rho_{\vec{h}}$ exists and is unique for any $\vec{h}>0$ (see also \cite{hjort1996locally} and discussion in \cite{otneim2016locally}). By letting $\vec{h}=\vec{h}_n\rightarrow0$, at an appropriate rate (see Assumption \ref{ass:rate}), the local correlation in the expression above, as mentioned in the previous section, satisfies 
\begin{equation}
\label{eq:population}
\psi(\vec{z}, \rho_0(\vec{z})) = f_{ij}(\vec{z}),
\end{equation} 
and we require the population value $\rho_0(\vec{z})$ to satisfy (\ref{eq:population}), cf. \cite{hjort1996locally} and \cite{tjostheim2013local}. Assuming (\ref{eq:population}) is not enough to ensure uniqueness of $\rho_0$ just by itself, though, even in our restricted case with $f_{ij}$ having standard normal margins, and the expectations and standard deviations of $\psi(\cdot, \rho)$ being equal to zero and one respectively. Consider for example the case where $f_{ij}$ is the bivariate Gaussian distribution with correlation coefficient $\rho^*\neq0$. It is obvious that $\rho_0(\vec{z})=\rho^*$ is the population parameter, but in the point $\vec{z}=\vec{0}$, we see that $\rho_0 = -\rho^*$ also satisfies (\ref{eq:population}). In this and more general situations, such problems are avoided by approximating with a Gaussian in successively smaller neighborhoods. We must therefore make the following assumption that guarantees a well defined population parameter at the point $\vec{z}$:
\begin{ass}\label{ass:existence}
For any sequence $\vec{h}_n$ tending to zero as $n\rightarrow\infty$ there exists for the bivariate marginally standard Gaussian vector $(Z_i, Z_j)$ a unique $\rho_{\vec{h}_n}(\vec{z})$ that satisfies (\ref{eq:limit}), and there exists a $\rho_0(\vec{z})$ such that $\rho_{\vec{h}_n}\rightarrow \rho_0(\vec{z})$.
\end{ass} 

See \cite{tjostheim2013local} for a discussion of Assumption \ref{ass:existence}, and see \cite{berentsen2016some} for a discussion of an alternative neighborhood-free approach to defining the population parameter by means of matching the partial derivatives of the locally Gaussian approximation with the true underlying density function. Assumption \ref{ass:existence} essentially ensures that we estimate the joint densities of each pair of transformed variables consistently, but the joint density $f_0(\vec{z}) = \Psi(\vec{z}, \vec{R}_0)$, where $\vec{R}_0 = \{\rho_{0,ij}(z_i, z_j)\}_{i<j}$, and $\Psi(\cdot, \vec{R})$ is the standardized multivariate Gaussian density function with correlation matrix $\vec{R}$, is not necessarily equal to the true density of the standardized variables, which we for simplicity denote by $f(\vec{z})$. For this to be true, $f(\vec{z})$ must be on the form 
\begin{equation} \label{eq:popform}
f(\vec{z}) = \Psi(\vec{z}, \vec{R}_0),
\end{equation} 
and this is a restriction of a general density because  the entire dependence structure must be contained in the pairwise correlation functions $\rho_{0, ij}(z_i, z_j)$, which is true for distributions with the Gaussian copula (for which the correlation functions are constant in \emph{all} directions), or a stepwise Gaussian distribution as described by \cite{tjostheim2013local}, but it is difficult (but not paramount for our estimation procedure) to find more analytic examples. 

The class of density functions satisfying (\ref{eq:popform}), $H(f_0)$ say, is much richer than the Gaussian case, however, and our performance in estimating a given unconditional density $f(\cdot)$ is clearly sensitive to the distance from $f(\cdot)$ to its best approximant in $H(f_0)$. 

Imposing a sparsity requirement like (\ref{eq:simplification}) can be viewed in one of two ways. First, as a modeling assumption that can be formally tested, and then discarded if the test should fail. On the other hand, it can be viewed as a simplification of reality that arises due to computational necessity, much like additivity in non-parametric regression as explained in Section \ref{sec:LGDE}. We focus on the latter interpretation, and so the method must therefore be judged first and foremost by its performance in practical situations, like those being presented in Section \ref{sec:examples}. We also refer to \cite{otneim2016locally} for comprehensive simulations and discussions.

Next, we introduce time series dependence. A strictly stationary series of stochastic variables $\{X_n\}, n=1,2,\ldots$ is said to be $\alpha$-mixing if $\alpha(m)\rightarrow0$, where
\begin{equation}\label{eq:mixing}
\alpha(m) = \sup_{A\in\mathcal{F}_{-\infty}^0, B\in \mathcal{F}_m^{\infty}}|P(A)P(B) - P(AB)|,
\end{equation}
and where $\mathcal{F}_i^j$ is the $\sigma$-algebra generated by $\{X_m, i\leq m\leq j\}$ \citep[p. 68]{fan2003nonlinear}. We require the mixing coefficients (\ref{eq:mixing}) of our observations to tend to zero at an appropriate rate, which means that we can turn to standard theorems in order to establish the asymptotic properties of our estimator.   
\begin{ass}\label{ass:alpha}
For each pair $(i,j)$, $1\leq i\leq p$, $1\leq j\leq p$, $i\neq j$, $\{(Z_i, Z_j)\}_n$ is $\alpha$-mixing with the mixing coefficients satisfying $\sum_{m\geq1}m^{\lambda}\alpha(m)^{1-2/\delta}<\infty$ for some $\lambda>1-2/\delta$ and $\delta > 2$.
\end{ass}
The next assumption links allowable bandwidth rates with the mixing rate:
\begin{ass}\label{ass:rate}
$n\rightarrow\infty$, and each of the bandwidths $h$ tend to zero such that $nh^{\frac{\lambda+2-2/\delta}{\lambda+2/\delta}}=O(n^{\epsilon_0})$ for some constant $\epsilon_0>0$. 
\end{ass}
In the current context $\{(Z_i, Z_j)\}_n$ is a bivariate process with standard normal margins. In the statement of Theorem \ref{thm:density}, Assumption \ref{ass:alpha} means that the general $p$-variate observations $\{\vec{X}_n\}$ are $\alpha$-mixing with the specified convergence rate for the mixing coefficients. This distinction has no practical importance when transforming back and forth between these two scales, because the mixing properties of a process are conserved under any measurable transformation \citep[p. 69]{fan2003nonlinear}. 

We need a compact parameter space and some regularity conditions on the kernel function in order to prove consistency and asymptotic normality for the local correlations:
\begin{ass} \label{ass:compactness}
The parameter space $\Theta$ for $\rho$ is a compact subset of $(-1,1)$.
\end{ass}
\begin{ass}
\label{ass:kernel}
The kernel function satisfies $\sup_{\vec{z}}|K(\vec{z})|\allowbreak <\infty$, $\int |K(\vec{y})|\,\textrm{d}\vec{y}\allowbreak<\infty$, $\partial/\partial z_i K(\vec{z})<\infty$ and $\lim_{z_{i}\rightarrow\infty}\allowbreak|z_iK(z_i)| = 0$ for $i=1,2$.
\end{ass}
\begin{thm} \label{thm:consistency} 
Let $\{(Z_i, Z_j)\}_n$ be identically distributed bivariate stochastic vectors with standard normal margins. Denote by $\rho_0(\vec{z})$ the local Gaussian correlation between $Z_i$ and $Z_j$, and by $\widehat \rho_n(\vec{z})$ its local likelihood estimate. Then, under assumptions \ref{ass:existence}-\ref{ass:kernel}, $\widehat\rho_n(\vec{z}) \stackrel{P}{\rightarrow}\rho_0(\vec{z})$ as $n\rightarrow\infty$.
\end{thm}
\begin{proof}
See Appendix \ref{app:proofconsistency}.
\end{proof}

\citet[pp. 76-77]{fan2003nonlinear} provide a general central limit theorem for non-parametric regression. It is applicable to the local correlations, with obvious adaptations in order to achieve consistent notation. Assume now that $\{\vec{Z}_n\}$ is a sequence of $p$-variate observations having standard normal margins, and denote by $\vec{\rho} = (\rho_1, \ldots,\allowbreak \rho_{p(p-1)/2})$ the vector of local correlations, which has one component for each pair of variables. The local correlations are estimated one by one using the scheme described above, and denote by $\widehat{\vec{\rho}}$ the estimate of $\vec{\rho}$. Further, as all bandwidths are assumed to tend to zero at the same rate, statements like $h^2$ are taken to mean the product of any two bandwidths $h_i$ and $h_j$.

The local correlation estimates are then jointly asymptotically normal:

\begin{thm} \label{thm:anorm}
Under assumptions \ref{ass:existence}-\ref{ass:kernel},
$$\sqrt{nh_n^2}\left(\widehat{\vec{\rho}}_n - \vec{\rho}_0\right) \stackrel{\mathcal{L}}{\rightarrow} N(\vec{0}, \vec{\Sigma}),$$ 
where $\vec{\Sigma}$ is a diagonal matrix with components
$$\vec{\Sigma}^{(k,k)} = \frac{f_k(\vec{z}_k)\int K^2(\vec{y}_k)\,\textrm{d}\vec{y}_k}{u^2(\vec{z}_k,\rho_{0,k}(\vec{z}_k))\psi^2(\vec{z}_k,\rho_{0,k}(\vec{z}_k))},$$
where $k=1,\ldots,p(p-1)/2$ runs over all pairs of variables, $f_k$ is the corresponding bivariate marginal density of the pair $\vec{Z}_k$, $\psi(\cdot)$ is defined in (\ref{eq:gaussfam}) and $u(\cdot)$ is defined in the paragraph following equation (\ref{eq:limit}).
\end{thm}
When comparing with the corresponding result in \cite{otneim2016locally}, we see that the mixing has no effect on the asymptotic covariance matrix compared with the iid case. See Appendix \ref{app:proofanorm} for proof. 

The preceding theorems lead up to the following asymptotic result for the locally Gaussian conditional density estimates, which is analogous to the corresponding result in \cite{otneim2016locally} in the unconditional case. Denote by $f_0(\vec{x}_1|\vec{X}_2 = \vec{x}_2)$ the locally Gaussian conditional density function of $\vec{X}_1|\vec{X}_2=\vec{x}_2$ (where $\vec{X} = (\vec{X}_1,\vec{X}_2)$ does not necessarily have standard normal marginals), which is obtained by replacing $f_{\vec{Z}}/f_{\vec{Z}_2}$ with $\Psi^*(\vec{z};\vec{\mu}^*_0, \vec{\Sigma}^*_0)$ in equation (\ref{eq:backtrans}). The parameters $\vec{\mu}^*_0$ and $\vec{\Sigma}^*_0$ are again obtained from formulas (\ref{eq:locmean}) and (\ref{eq:locvar}) using the population values of the local correlations as defined in Assumption \ref{ass:existence}. 

Following the algorithm in Section \ref{sec:estimation}, we must estimate the local Gaussian correlation for pairs of variables $\widehat{\vec{Z}}_n = \{(\widehat Z_i, \widehat Z_j)\}_n$ as defined in equation (\ref{eq:pseudo}), that are not exactly marginally standard normal, because the distribution functions $F_i(\cdot)$, $i=1,\ldots,p$ must be estimated from the data. In the same way as for the iid case in \cite{otneim2016locally}, we need some extra assumptions on the pairwise copulas between the components in $\vec{X}$ to ensure that using the empirical distribution distribution functions instead of the true distributions will not affect the asymptotic distribution of the LGDE conditional density estimate. The following assumptions are taken directly from \cite{geenens2014probit}, who derive the asymptotic properties of a local likelihood copula density estimator in the bivariate case, that is also based on transformations to marginal standard normality. 

\begin{ass} \label{ass:geenens1}
The marginal distribution functions $F_{1}, \ldots,F_{p}$ are strictly increasing on their support.
\end{ass}

\begin{ass} \label{ass:geenens2}
Each pairwise copula $C_{ij}$ of $(X_i,X_j)$ is such that $(\partial C_{ij}/\partial u)(u,v)$ and $(\partial^2 C_{ij}/\partial u^2)(u,v)$ exist and are continuous on $\{(u,v):u\in(0,1), v \in [0,1]\}$, and $(\partial C_{ij}/\partial v)(u,v)$ and $(\partial^2 C_{ij}/\partial v^2)(u,v)$ exist and are continuous on $\{(u,v):u\in[0,1], v \in (0,1)\}$. In addition, there are constants $K_i$ and $K_j$ such that
\begin{align*} 
\left|\frac{\partial^2C_{ij}}{\partial u^2}(u,v)\right| &\leq \frac{K_i}{u(1-u)} & \textrm{ for } (u,v) \in (0,1)\times[0,1], \\
\left|\frac{\partial^2C_{ij}}{\partial v^2}(u,v)\right| &\leq \frac{K_j}{v(1-v)} & \textrm{ for } (u,v) \in [0,1]\times(0,1).
\end{align*}
\end{ass}

\begin{ass} \label{ass:geenens3}
Each density $c_{i,j}$ of $C_{i,j}$ exists, is positive, and admits continuous partial derivatives to the fourth order on the interior of the unit square. In addition, there is a constant $K_{00}$ such that
$$
c(u,v) \leq K_{00}\min\left(\frac{1}{u(1-u)},\frac{1}{v(1-v)}\right)  \textrm{ for all }(u,v)\in(0,1)^2.
$$
\end{ass}  

These smoothness assumptions are quite weak, as can be seen from the discussion in \cite{geenens2014probit}. Finally, we need to assume that the final back-transformation of the density estimate converge faster than the nonparametric rate of $\sqrt{nh^2}$:

\begin{ass} \label{ass:backtrans}
The estimates of the marginal densities and quantile functions that are used for the back-transformations in (\ref{eq:backtrans}), are asymptotically normal with convergence rates faster than $\sqrt{nh^2}$.
\end{ass}
As we use the logspline-estimator \citep{stone1997polynomial} for the back-transformations in all our examples, we discuss its large sample properties in light of assumption \ref{ass:backtrans} in Appendix \ref{app:logspline}. Another possible candidate is the basic univariate kernel estimator, which, under some regularity conditions, converges as $\sqrt{nh}$.

\begin{thm} \label{thm:density}
Let $\{\vec{X}_n\}$ be a strictly stationary process with density function $f_{\vec{X}}(\vec{x})$. Partition $\vec{X}$ into $\vec{X}_1 = (X_1, \ldots, X_k)$ and $\vec{X}_2 = (X_{k+1},\ldots,X_p)$, and let $\widehat f_{\vec{X}_1|\vec{X}_2}\allowbreak(\vec{x}_1|\vec{X}_2 = \vec{x_2})$ be the estimate of the conditional density $f_{\vec{X}_1|\vec{X}_2}$ that is obtained using the procedure in Section \ref{sec:estimation}. Then, under assumptions \ref{ass:existence}-\ref{ass:backtrans}, 
\begin{align*}
&\sqrt{nh_n^2}\left(\widehat f_{\vec{X}_1|\vec{X}_2}(\vec{x}_1|\vec{X}_2 = \vec{x_2}) - f_0(\vec{x}_1|\vec{X}_2 = \vec{x_2})\right) \\
& \qquad\qquad\qquad \stackrel{\mathcal{L}}{\rightarrow} N \big( 0, \psi^*(\vec{z};\vec{\mu}^*_0, \vec{\Sigma}^*_0)^2g(\vec{x})^2 \vec{u}^T(\vec{z}; \vec{\mu}^*_0, \vec{\Sigma}^*_0)\, \vec{\Sigma}\, \vec{u}(\vec{z}; \vec{\mu}^*_0, \vec{\Sigma}^*_0),
\end{align*}
where
\begin{align*}
g(\vec{x}) &= \prod_{i=1}^k f_i(x_i)/\phi(z_i), \\
\vec{z} &= \{z_i\}_{i=1,\ldots,p} = \{\Phi^{-1}(F_i(x_i))\}_{i=1,\ldots,p},
\end{align*}
and $\vec{u}(\vec{z}) = \nabla \log\psi^*(\vec{z}, \vec{\mu}^*_0, \vec{\Sigma}^*_0)$, where the gradient is taken with respect to the vector of local correlations.
\end{thm}
See Appendix \ref{app:densityproof} for a proof.

\begin{figure*}[t]
\includegraphics[width = \textwidth]{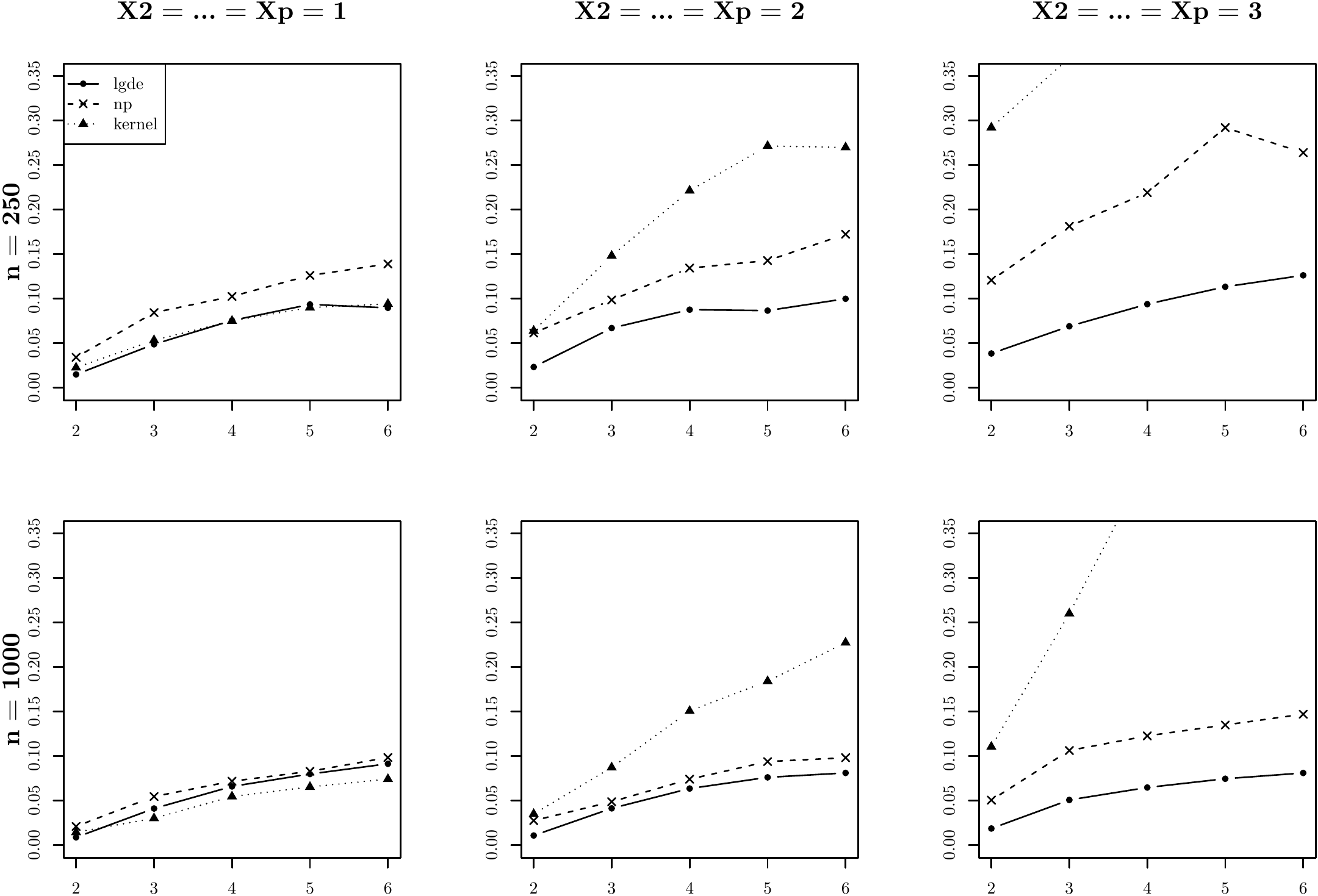}
\caption{The integrated squared error of conditional density estimates of $f_{X_1|X_2,\ldots,X_p}$ as a function of $p$, generated from a density with exponential margins and a Joe copula with Kendall's Tau equal to 0.6.}
\label{fig:expjoe}
\end{figure*}

\begin{figure*}[t]
\includegraphics[width = \textwidth]{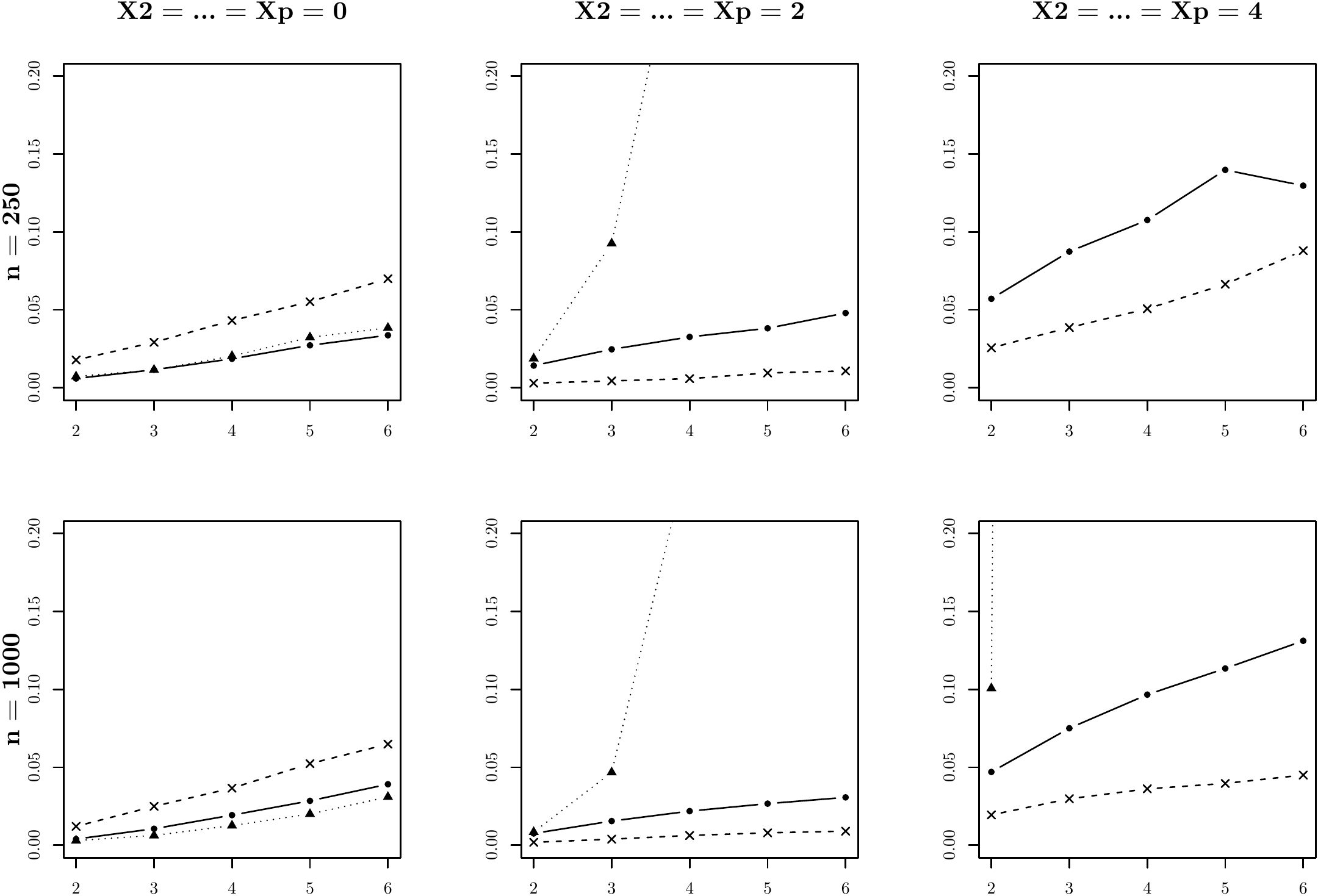}
\caption{The integrated squared error of conditional density estimates of $f_{X_1|X_2,\ldots,X_p}$ as a function of $p$, generated from the multivariate $t$-distribution with 4 degrees of freedom.}
\label{fig:t4}
\end{figure*}

\begin{figure*}[t]
\includegraphics[width = \textwidth]{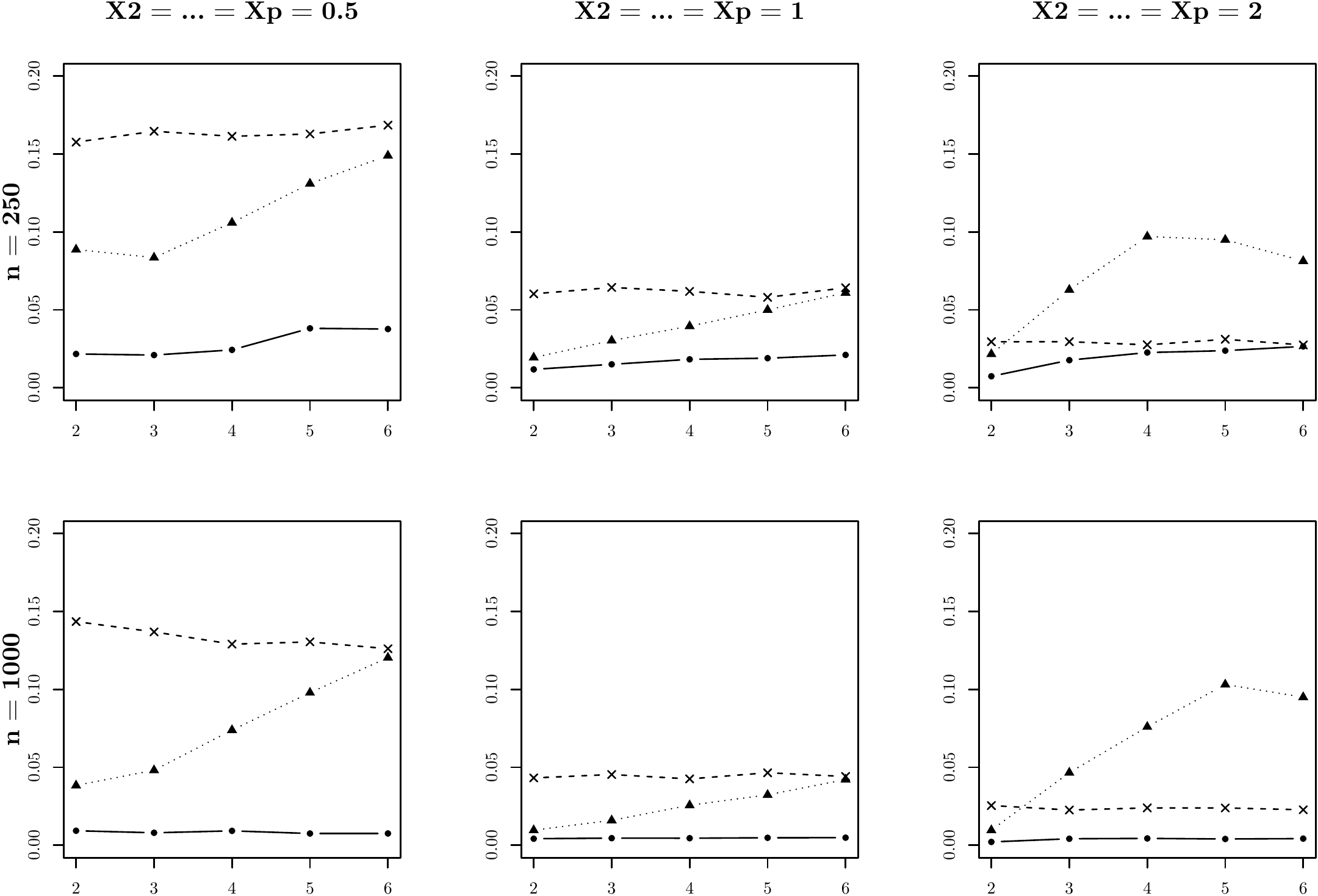}
\caption{The integrated squared error of conditional density estimates of $f_{X_1|X_2,\ldots,X_p}$ as a function of $p$, generated from a density in which the first two variables are marginally log-normal with a $t$(10)-copula, and the rest of the variables are multivariate $t$(5)-distributed, independently from $(X_1, X_2)$.}
\label{fig:lnormttindep}
\end{figure*}

\section{Examples} \label{sec:examples}
The asymptotic results of the preceding section will not give us the complete picture on how the LGDE estimator of conditional densities behaves in practice for a finite sample. We must also take into account that the simplification (\ref{eq:simplification}) of the dependence structure could introduce an approximation error in practical applications, the size of which depends on the problem at hand. We proceed to apply our new estimator to a series of problems using real and simulated data, and compare it with existing methods.

It is customary in the copula literature to generate pseudo-observations by means of the marginal empirical distribution functions, and this is why we can prove Theorem \ref{thm:density} by mostly referring to existing results. The back-transformation (\ref{eq:backtrans}) must be smooth and invertible, making a standard marginal kernel estimate a natural choice. Extensive testing, however, has revealed that we obtain better finite sample performance if we use the logspline method by \cite{stone1997polynomial} for marginal density and distribution estimates, not only in the back-transformation (\ref{eq:backtrans}), but also in generating the marginally Gaussian pseudo-observations (\ref{eq:pseudo}). The following examples, as well as the computer code that accompany this article as supplementary material, therefore use the logspline estimator for both of these purposes. We argue in Appendix \ref{app:logspline} that the asymptotic properties of the logspline estimator do not change when applied to $\alpha$-mixing data compared to independent data.    
 
\subsection{Conditional density estimation}
\label{subseq:densest}
\begin{figure*}[t]
\includegraphics[width = \textwidth]{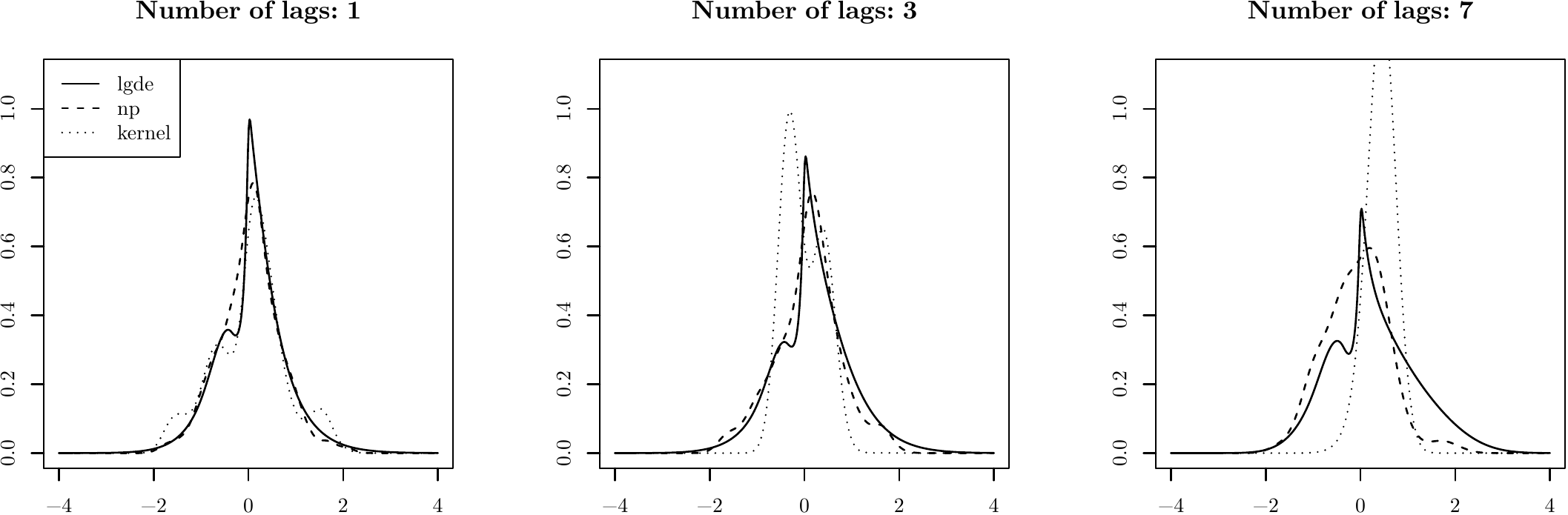}
\caption{Estimate of the conditional density of the US log-returns conditioned on $X_{t-1} = \cdots = X_{t-k} = -1$ with $k=1,3,7$ respectively.}
\label{fig:varsingle}
\end{figure*}

\subsubsection{Simulated data with relevant variables}
In this section, we wish to investigate the sensitivity of various methods with respect to the number of explanatory variables in the problem, and begin by presenting some simulation experiments in which we generate data from test distributions, measure the integrated squared error (ISE) of our conditional density estimate, and compare it with the two natural competitors which are readily available for implementation: the na\"{i}ve approach, where the numerator and denominator of (\ref{eq:cond}) are estimated separately using the multivariate kernel estimator with the plug-in bandwidth selector of \cite{wand1994multivariate}, and the specialized kernel method by \cite{li2007nonparametric}, which we denote by the name of the software package written in the R programming language \citep{R} from which it can be calculated: ``NP'' \citep{hayfield2008nonparametric}.  

The first test distribution has standard exponentially distributed margins, and the dependence structure is defined by the Joe copula (see e.g. \citet[p. 116, distribution 6]{nelsen2013introduction}) with parameter $\theta = 3.83$, which corresponds to a Kendall's Tau of 0.6 between all pairs of variables. For each dimension $p$, ranging from 2 to 6, we generate $2^7 = 128$ data sets, and estimate the conditional density of $X_1|X_2 = \cdots = X_p = c$, with $c$ being equal to 1,2 and 3 in this example. We calculate the ISE of the density estimates numerically over 2000 equally spaced grid points, and graph the mean of the estimated errors as a function of the dimension for two different sample sizes ($n=250$ and $n=1000$), see Figure \ref{fig:expjoe}. 

The basic kernel estimator performs well in the center of the distribution, especially in the example with sample size 1000. When we condition on values that are farther out in tail, however, it quickly deteriorates as the dimension increases. This behavior is of course expected because of the curse of dimensionality. The NP-estimator is clearly a major improvement to na\"{i}ve kernel estimation of conditional densities, but in this example we see that the LGDE approach is the overall best performer. It matches the purely non-parametric methods in lower-dimensional cases, but also boasts a greater robustness against increasing dimensionality than its competitors. The tail behavior of the LGDE is much better than the other two methods. It is governed by a Gaussian distribution, which again is determined locally by the behavior of $f_{X_1|X_2,\ldots,X_p}$ in the tail.

\subsubsection{Simulated data from a heavy-tailed distribution}
\cite{otneim2016locally} show that the unconditional version of the LGDE does not work very well when fitted to the heavy-tailed $t(4)$-distribution. The reason for this is not entirely clear, but one explanation is that the cross-validated bandwidths are too small. The conditional version of the LGDE also starts to struggle when presented with data from this distribution, as can be seen in Figure \ref{fig:t4}. It is expected that using the $t$-distribution in the same pairwise and local manner as we use the Gaussian distribution here, will improve this fit, and we discuss this more closely in Section \ref{sec:conclusion}. The conditional density estimator by \cite{li2007nonparametric} is the best alternative in this case if the explanatory variables are not in the center of the distribution. 

\subsubsection{Simulated data with irrelevant variables}
One challenge in estimating conditional densities is to discover, and take account of, independence between variables. We have not addressed this problem explicitly in the derivation of our estimator, contrary to the NP-estimator by \cite{li2007nonparametric}, which smooths irrelevant variables away automatically. In our next example, however, most of the explanatory variables are independent from the response variable, but they are mutually dependent themselves. In the two-dimensional case with $\vec{X} = (X_1, X_2)$, we generate data from a bivariate distribution with log-normal margins that has been assembled using the $t$-copula with 10 degrees of freedom. For all dimensions greater than two, the remaining variables $X_3, \ldots, X_p$ are drawn from a multivariate $t$-distribution with 5 degrees of freedom, but independent from $(X_1, X_2)$. 

It turns out that our approach handles this case very well, see Figure \ref{fig:lnormttindep}. None of the methods have errors that grow sharply with the dimension, which indicate that they more or less ignore the extra noise that the extra dimensions contains. The LGDE-method is clearly the best, however, according to this particular choice of error measure. The explanation for this is the equivalence between independence and the local correlation being equal to zero between marginally Gaussian variables, which in turn means that, by construction, variables that are independent from the response variable will have very little influence in the final conditional density estimate. 

\subsubsection{Real data with irrelevant variables}\label{subsubsec:irr}
We can explore this property using a real data set as well. Consider a subset of the data set which is also analyzed in \cite{otneim2016locally} comprising daily log-returns on the S\&P 500 stock index observed on 1443 days from January 3rd 2005 until July 14th, 2010. In this example we will use only the first 500 observations, so the financial crisis of 2008 is not included in this particular analysis. 

We know that there is very little extra information given the first lag in this time series, thus estimating the marginal density of these log returns by conditioning on more and more lags will not introduce more information, but rather noise, that should ideally be ignored by the estimation routine. 

Figure \ref{fig:varsingle} displays the marginal density estimates of the data, calculated using the three competing methods and conditioned on the preceding 1, 3 and 7 days' values respectively being equal to $-1$. All methods perform similarly in the first case in which we condition on only one variable. In the second panel we condition on three lags, which amounts to a four dimensional problem in terms of density estimation, and the na\"{i}ve kernel estimator, not surprisingly, struggles in this case. The other two methods, however, the NP and the LGDE, remain largely unchanged, which indicates that they, for the most part, ignore the additional two variables of data. When conditioning on 7 lags, the kernel estimator should not be trusted. The NP-estimator also appears to loose some characteristics, like the sharpness of its peak and the fatness of its right tail. The LGDE, on the other hand, seems to be the better performer in this case. Although the estimate is slightly deformed compared to the other two figures, its main characteristics are conserved. The tails in particular shows great robustness compared to the other two methods, and we believe that this behavior to a large part explains its good performance in simulation experiments, and we will also exploit this feature in Section \ref{subsec:var}.   

\subsubsection{Melbourne temperature data: comparison with local polynomials}
\begin{figure*}[!t]
\includegraphics[width = \textwidth]{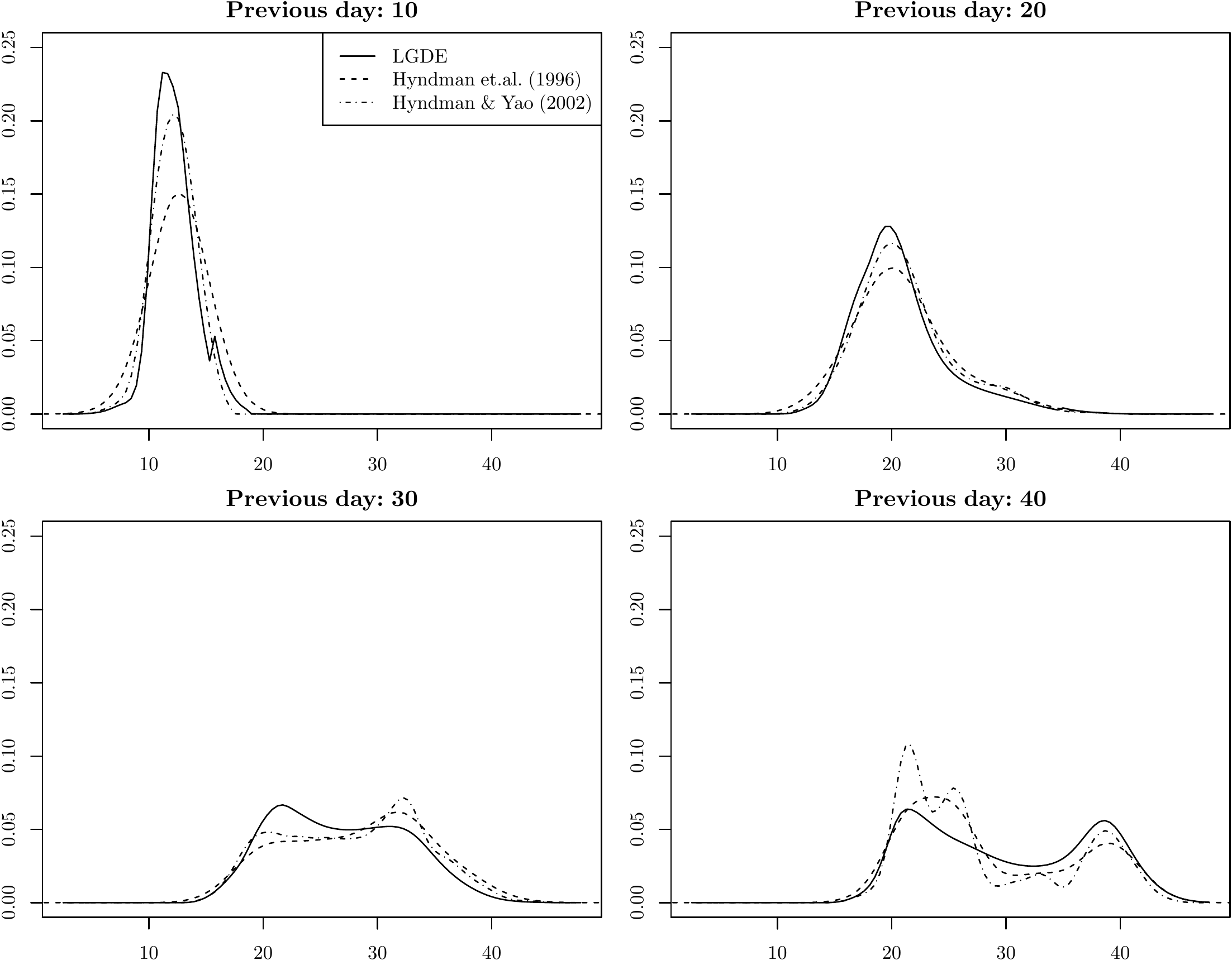}
\caption{Australian temperature data, with estimated conditional density of the maximum daily air temperature, given a preceding recording of 10, 20, 30 and 40 degrees Celsius respectively.}
\label{tempdata}
\end{figure*}
The local polynomial conditional density estimators of \cite{hyndman1996estimating} and \cite{hyndman2002nonparametric} is in its current implementation restricted to the case where the explanatory and response variables are both scalar, and is therefore not included in the simulation experiments of the preceding subsection. We will, however, compare these estimators to our approach using the Melbourne temperature data that is presented by \cite{hyndman1996estimating}. The data consists of daily recordings of the maximum air temperature in Melbourne, Australia from 1981 until 1990. It is known that a low maximum temperature one day most often results in a similar temperature the next day. Local meteorological conditions, however, have the effect that a high maximum temperature is often followed by either a large, or a much smaller observation, making the corresponding conditional density bimodal. The \cite{hyndman1996estimating}-estimator, which in this example is a local polynomial of order zero, recovers this phenomenon nicely, and although our locally Gaussian estimator is not identical, it gives a similar picture, see Figure \ref{tempdata}. The \cite{hyndman2002nonparametric}-estimator is a locally quadratic polynomial, and mostly agrees with the other methods, but seems to be slightly overfitting the density in the lower right panel. 

It is interesting to note that the bimodality of the LGDE-estimator is mirrored compared with the local polynomials in the lower left panel. 

\subsection{Partial correlation and covariance}

\begin{figure*}[t]
\includegraphics[width = \textwidth]{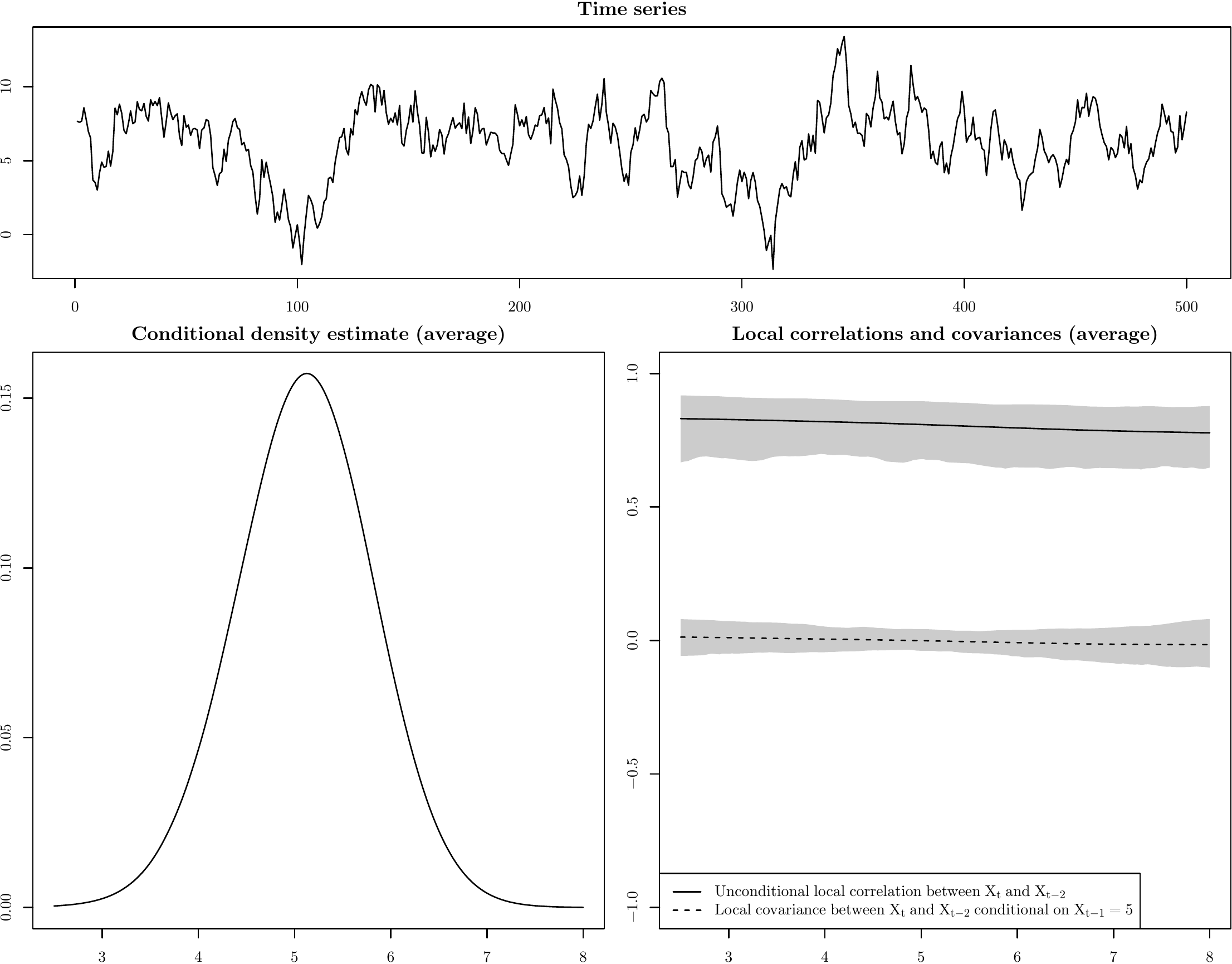}
\caption{The top panel displays a simulated time series. The lower left panel displays the average of the estimated conditional densities of $X_t|X_{t-2}=5$, and the lower right panel shows the unconditional diagonal local correlation between $X_t$ and $X_{t-2}$, as well as the same quantity when conditioned on the intermediate value $X_{t-1}$, with 95\% empirical confidence intervals.}
\label{fig:timeseries-sim}
\end{figure*}

The partial autocorrelation function for a stationary time series at lag $k$ is the correlation between $X_t$ and $X_{t-k}$, given the values of the intervening lags \citep[p. 98]{brockwell2013time}. The concept of partial correlation is very important, especially in the analysis of conditional dependencies in Bayesian networks. Partial local correlation is a natural extension of local correlation in light of the new theory allowing for dependent observations. Consider for example the nonlinear AR(1) model
$$X_t = 0.8X_{t-1} + 0.5\sqrt{|X_{t-1}|} + Z_t,$$
where the $Z_t$s are independent standard normal innovations. One realization of length 500 is plotted in the upper panel of Figure \ref{fig:timeseries-sim}. There is strong serial dependence in this model. Indeed, if we estimate the joint density of the lagged values $X_t$ and $X_{t-2}$ using the LGDE methodology, the estimated local correlation is close to 1. This can be seen in the lower right panel of Figure \ref{fig:timeseries-sim}, in which the local correlation for 300 realizations has been averaged and plotted as a solid line along the diagonal $x_t=x_{t-2}$, along with the empirical 95\% confidence interval. We do know from the Markov property of $\{X_t\}$, however, that $X_t$ is independent of $X_{t-2}$ given $X_{t-1}$, and this is clearly reflected in the estimated local covariance between the two variables for the joint $conditional$ density of $(X_t, X_{t-2})|X_{t-1}=x_{t-1}$ (where $x_{t-1}=5$ in this particular case), that has been plotted as a dashed line. We use the term local covariance here, instead of local correlation, because the diagonal elements in $\vec{\Sigma}$ as defined by (\ref{eq:locvar}) are no longer 1. As seen in the lower right panel of Figure \ref{fig:timeseries-sim}, the local covariance practically vanishes when the intermediate variable is conditioned upon.

The average of the estimated conditional densities in question has been plotted along its diagonal in the lower left panel of Figure \ref{fig:timeseries-sim}. 

\subsection{Forecasting the value-at-risk of a portfolio}
\label{subsec:var}
\begin{figure*}[t]
\includegraphics[width = \textwidth]{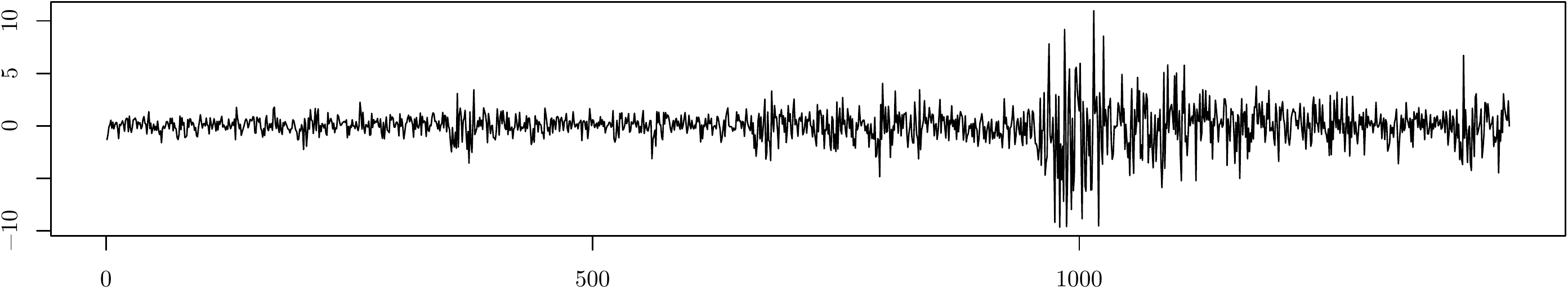}
\caption{Value of the portfolio over a period of 1442 days.}
\label{fig:timeseries}
\end{figure*}
\begin{table} 
\centering
\caption{Proportion of observations exceeding the estimated VaR} 
\begin{tabular}{lccc} 
\hline \noalign{\smallskip}
& \multicolumn{3}{c}{Level} \\
\noalign{\smallskip}\hline
\noalign{\smallskip}
Method & 0.005 & 0.01 & 0.05 \\
\noalign{\smallskip}\hline 
\noalign{\smallskip}
LGDE & 0.014 & 0.017 & 0.072 \\
np   & 0.084 & 0.097 & 0.161 \\
Kernel & 0.117 & 0.134 & 0.187 \\
Gaussian & 0.045 & 0.064 & 0.125 \\
\label{tab:VaR}
\end{tabular} 
\end{table}
   
There is a vast literature available on portfolio optimization theory. A vital element when selecting the optimal distribution of wealth over a set of assets is the estimation of risk, of which the Value-at-Risk (VaR) is a common measure. The VaR of a portfolio at level $\alpha$ is simply the upper $(1-\alpha)$-quantile of the loss-distribution of the portfolio, which usually needs to be estimated from past data. 

We look at the S\&P 500 data from Section \ref{subsubsec:irr}, as well as the corresponding log-returns on the British FTSE 100 index and the Norwegian OBX, and consider the observations on all 1443 days. In this toy example, we will show that our conditional density estimator may well be used as an instrument in estimating the VaR.

We wish to estimate the daily VaR of a portfolio consisting of each of these indices, equally weighted, conditioned on the observed log-returns on preceding days. The log-returns of this portfolio is plotted in Figure \ref{fig:timeseries}. Denote by $(X_1, \ldots, X_4)$ the four-dimensional vector that we observe each day, in which $X_1$ is the value of the portfolio that day, and $X_2,\ldots, X_4$ are the values of its individual components on the preceding day. On each day we estimate the conditional density of $X_1|X_2=x_2, \ldots, X_4 = x_4$ and calculate the $\alpha$-level VaR by numerical integration. We do the same by using the non-parametric kernel estimator by \cite{li2007nonparametric}, naive kernel estimator, as well as by assuming the data to be jointly Gaussian and calculating the quantile from a fully parametric fit. We start our analysis on day number 500, and for computational feasibility, we calculate the bandwidths for all methods on the first day of analysis only, and keep them constant throughout the period. 

Table \ref{tab:VaR} displays the result of our analysis. For each method we count the proportion of observations that exceed the estimated VaR on the corresponding day. We see that all methods under-estimate the risk, but the LGDE-approach is clearly the better performer, which we believe is due to its tendency to allow fat tails in the density estimates, see e.g. Figure \ref{fig:varsingle}, even though it has a \emph{local} Gaussian tail. 

A thorough treatment of this topic would include pre-filtering of the data using for example a GARCH-type model as found in \cite{palaro2006using}, as well as implementation of the LGDE in optimization over the portfolio weights, but that is beyond the scope of this paper.

\section{Conclusion and further work} \label{sec:conclusion}
Constructing non-parametric estimates of conditional density functions is a fundamental problem in statistics, but it is difficult, because many of the existing methods rely either on the traditional kernel density estimator, or on separate estimates of the numerator and denominator in the definition of the conditional density, or, most often, both. This could work in lower dimensional problems, especially if we keep ourselves away from the tails of the distribution in question.

We have shown, however, that by using the LGDE methodology, both of these problems tend to disappear. The simplified locally Gaussian estimates cope far better in higher dimensions than the kernel estimator, and it provides an explicit expression of the conditional density estimates, without the need for separate estimates of the numerator and denominator. The result is a general conditional density estimator for continuous data that is robust against dimensionality issues, modeling error, as well as noise induced by irrelevant variables.

These properties have been demonstrated through examples and asymptotic derivations. A more comprehensive theoretical analysis of the LGDE-framework and its possible generalizations remains to be developed, and will be the subject of later studies. For example, the degree to which a general multivariate density function can be characterized by pairwise locally Gaussian correlations, or the distance between $f(\vec{x})$ and $f_0(\vec{x})$ in keeping with the notation from Section \ref{sec:asymp}, is a challenge, cf. \cite{otneim2016locally}. Further, if the LGDE-approach can be labeled as a two-fold approximation compared to the fully non-parametric, or $p$-fold, estimation procedure in which we omit the simplification (\ref{eq:simplification}), it might be worthwhile to develop a general procedure allowing for a $k$-fold model, in which each local correlation depends on $k$ variables, with $k$ increasing, and these variables being selected based on data analogously to variable selection methods in regression. In theory, this can be generalized even further by replacing the normal distribution as a building block, with another member of the family of elliptical distributions that also organizes its parameters in a covariance-like matrix structure. Deriving conditional densities from such a general model requires more work, but should in principle be possible.   

\appendix 

\section{Proofs}

\subsection{Proof of Theorem \ref{thm:consistency}} \label{app:proofconsistency}
Except from a slight modification that accounts for the replacement of independence with $\alpha$-mixing, the proof of Theorem \ref{thm:consistency} is identical to the corresponding proof in \cite{otneim2016locally}, which again is based on the global maximum likelihood case covered by \cite{severini2000likelihood}. For each location $\vec{z}$, that we for simplicity suppress from notation, denote by $Q_{\vec{h}_n,K}(\rho)$ the expectation of the local likelihood function $L_n(\rho, \vec{Z})$. Consistency follows from uniform convergence in probability of $L_n(\rho, \vec{Z})$ towards $\allowbreak Q_{\vec{h}_n,K}\allowbreak(\rho)$, conditions for which are provided in Corollary 2.2 by \cite{newey1991uniform}. 

The result requires compact support of the parameter space, equicontinuity and Lipschitz continuity of the family of functions $\{Q_{\vec{h}_n, K}(\rho)\}$, as well as pointwise convergence of the local likelihood functions. Compactness is covered by Assumption \ref{ass:compactness}, and the demonstration of equi- and Lipschitz continuity in \cite{otneim2016locally} does not rely on the independent data assumption. Pointwise convergence follows from a standard non-parametric law of large numbers in the independent case. Our assumption \ref{ass:alpha} about $\alpha$-mixing data, however, ensures that pointwise convergence still holds, see for example Theorem 1 by \cite{irle1997consistency}, conditions for which are straightforward to verify in our local likelihood setting.

The rest of the proof is identical to the corresponding argument by \cite[pp. 105-107]{severini2000likelihood}.

\subsection{Proof of Theorem \ref{thm:anorm}} \label{app:proofanorm}
Consider first the bivariate case, in which there is only one local correlation to estimate. The first part of the proof goes through exactly as in the iid-case of \cite{otneim2016locally}. We follow the argument for global maximum likelihood estimators as presented in Theorem 7.63 by \cite{schervish1995theory}. The statement of Theorem \ref{thm:anorm} follows provided that 
\begin{equation}\label{oneparameter}  
Y_n(\vec{z}) = \sum_{i=1}^nK\left(|\vec{h}_n|^{-1}(\vec{Z}_i - \vec{z})\right)u(\vec{Z}_i,\rho_0) = \sum_{i=1}^nV_{ni}, 
\end{equation} 
is asymptotically normal, and this follows from a standard Taylor expansion. In the iid-case, the limiting distribution of (\ref{oneparameter}) is derived using the same technique as when demonstrating asymptotic normality for the standard kernel estimator, for example as in the proof of Theorem 1A by \cite{parzen1962estimation}. We establish asymptotic normality of (\ref{oneparameter}) in case of $\alpha$-mixing data, however, by going through the steps used in proving Theorem 2.22 in \cite{fan2003nonlinear}. Let $W_i = h^{-1}V_{ni}$, then
\begin{align*}
\frac{1}{nh^2}\textrm{Var}(Y_n(\vec{z}))
&= \frac{1}{nh^2} \left\{ \sum_{i=1}^n \textrm{Var}(V_{ni}) + 2\sum\sum_{1\leq i < j \leq n}\textrm{Cov}(V_{ni},V_{nj})\right\} \\
&= \textrm{Var}(W_1) + 2\sum_{j=1}^n (1-j/n)\textrm{Cov}(W_1, W_{j+1}),
\end{align*}
where
\begin{align*}
\textrm{Var}(W_1) &= \textrm{E}(W_1^2) - (\textrm{E}(W_1))^2 \\
&= \int h^{-2}u^2(\vec{z}, \rho_0)K^2(h^{-1}(\vec{y} - \vec{z}))f(\vec{y}) \,\textrm{d}\vec{y} + O(h^2) \\
&= \int u^2(\vec{z} + h\vec{v})K^2(\vec{v})f(\vec{z} + h\vec{v})\,\textrm{d}\vec{v} + O(h^2) \\
&\rightarrow u^2(\vec{z}, \rho_0)f(\vec{z})\int K^2(\vec{v})\,\textrm{d}\vec{v} \stackrel{\textrm{def}}{=} M(\vec{z}) \,\,\textrm{as}\,\,
\vec{h}\rightarrow0,
\end{align*}
and
$$
|\textrm{Cov}(W_1, W_{j+1})|=|\textrm{E}(W_1W_{j+1}) - \textrm{E}(W_1)\textrm{E}(W_{j+1})|  = O(h^2),
$$
using the same argument once again. Therefore,
$$\left|\sum_{j=1}^{m_n}\textrm{Cov}(W_1,W_{j+1})\right| = O(m_nh^2).$$
\cite{fan2003nonlinear} require that 
\begin{equation}\label{ass:finite}
\textrm{E}(u(\vec{Z}_n, \rho_0(\vec{z}))^{\delta})<\infty
\end{equation} 
for some $\delta>2,$ but this is of course true for our transformed data, because it is marginally normal. In proposition 2.5(i) by \cite{fan2003nonlinear} we can therefore use $p=q=\delta >2$ in order to obtain, for some constant $C$,

$$|\textrm{Cov}(W_|, W_{j+1})|\leq C\alpha(j)^{1-2/\delta}h^{4/\delta-2}.$$
Let $m_n=(h_n^2|\log h_n^2|)^{-1}$. Then $m_n\rightarrow\infty$, $m_nh^2\rightarrow0$, and
$$
\sum_{j=m_n+1}^{n-1}|\textrm{Cov}(W_1,W_{j+1})|\leq C\frac{h^{4/\delta-2}}{m_n^{\lambda}}\sum_{j=m_n+1}^nj^{\lambda}\alpha(j)^{1-2/\delta}\rightarrow0,
$$
which follows from assumption \ref{ass:alpha}. Thus,
$$\sum_{j=1}^{n-1}\textrm{Cov}(W_1, W_{j+1})\rightarrow0,$$
and it follows that 
$$\frac{1}{nh^2}\textrm{Var}(Y_n(\vec{z})) = M(\vec{z})(1+o(1)).$$

The proof now continues exactly as in \cite{fan2003nonlinear} using the "big block small block" technique, but with the obvious replacement of $h$ with $h^2$ to accommodate the bivariate case. 

We expand the argument to the multivariate case using the Cram\`{e}r-Wold device. Let $\vec{\rho} = (\rho_1, \ldots, \rho_d)^T$ be the vector of local correlations, where $d = p(p-1)/2$, write $\vec{u}(\vec{z}, \vec{\rho}_0) = (u_1(\vec{z}, \vec{\rho}_0), \ldots, u_d(\vec{z}, \vec{\rho}_0))$ and let $\vec{S}_{n}(\vec{z}) = \{S_{ni}(\vec{z})\}_{i=1}^d$, where
$$S_{ni} = \sum_{n=1}^nu_k(\vec{Z}_t, \vec{\rho}_0)K(|\vec{h}|^{-1}(\vec{Z}_t - \vec{z})).$$  
We must show that 
\begin{equation}
\label{cramerwold}
\sum_ka_kS_{nk} \stackrel{\mathcal{L}}{\rightarrow} \sum_ka_kZ_k^*,
\end{equation}
where $\vec{a} = (a_1, \ldots, a_d)^T$ is an arbitrary vector of constants, and $\vec{Z}^* = (Z_1^*, \ldots, Z_k^*)$ is a jointly normally distributed random vector. Because of Slutsky's Theorem, it suffices to show that the left hand side of (\ref{cramerwold}) is asymptotically normal. This follows from observing that it is on the same form as the original sequence comprising $S_n$, with
$$\sum_ka_kS_{nk} = \sum_nu^*(\vec{Z}_n, \vec{\rho}_0)K(|\vec{h}|^{-1}(\vec{Z}_n-\vec{z})),$$
where $u^*(\vec{Z}_n, \vec{\rho}_0) = \sum_ka_ku_k(\vec{Z}_n,\vec{\rho}_0)$. It is well known that any measurable mapping of a mixing sequence of random variables inherit the mixing properties of the original series, so condition \ref{ass:alpha} is therefore satisfied by the linear combination. The new sequence of observations satisfies (\ref{ass:finite}) because it follows from Jensen's inequality that for $\delta>2$,
\begin{align*}
\left[\frac{u^*(\vec{Z}_t, \vec{\rho}_0)}{\sum_ka_k}\right]^{\delta} &= \left[\frac{\sum_ka_ku_k(\vec{Z}_t, \vec{\rho}_0)}{\sum_ka_k}\right]^{\delta} \\&\leq \frac{\sum_ka_k[u_k(\vec{Z}_t,\vec{\rho}_0)]^{\delta}}{\sum_ka_k},
\end{align*}
so that
$$\textrm{E}[u^*(\vec{Z}_t,\vec{\rho}_0)]^{\delta} \leq \sum_ka_k\textrm{E}[u_k(\vec{Z}_t,\vec{\rho}_0)]^{\delta}\left[\sum_ka_k\right]^{\delta - 1}<\infty.$$

The off-diagonal elements in the asymptotic covariance matrix are zero using the same arguments as in \cite{otneim2016locally}.

\subsection{Proof of Theorem \ref{thm:density}}
\label{app:densityproof}
The key to proving \ref{thm:density} is to show that the asymptotic distribution of (\ref{oneparameter}) remains unchanged when the marginally standard normal stochastic vectors $\vec{Z}_n$ are replaced with the pseudo-observations 
$$\widehat{\vec{Z}}_n = \left(\Phi^{-1}(\widehat F_1(X_{j1})), \ldots, \Phi^{-1}(\widehat F_p(X_{jp}))\right)^T,$$
where $\widehat F_i(\cdot)$, $i=1,\ldots,p$ are the marginal empirical distribution functions. This is shown in the independent case under assumptions \ref{ass:geenens1}-\ref{ass:geenens2} in \cite{otneim2016locally}, by providing a slight modification to Proposition 3.1 by \cite{geenens2014probit}. The essence in that proof is the convergence of the empirical copula process, which remain unchanged if we replace the assumption of independent observations with $\alpha$-mixing, according to \cite{bucher2013empirical}.

The multivariate delta method states that if $\allowbreak\sqrt{nh^2}(\theta_n - \theta) \stackrel{\mathcal{L}}{\rightarrow} N(0, A)$ and $q:R^n\rightarrow R$ has continuous first partial derivatives, then $\sqrt{nh^2}(q(\theta_n) - q(\theta)) \stackrel{\mathcal{L}}{\rightarrow} \allowbreak N(0, \allowbreak \nabla q(\theta)^TA\nabla q(\theta))$ \citep[p. 403]{schervish1995theory}). In our case, $q(\vec{\rho}) = \Psi(\vec{z}, \vec{R})g(\vec{x})$, and $$\nabla q(\vec{\rho}) = \Psi(\vec{z}, \vec{R})g(\vec{x})\vec{u}(\vec{z}, \vec{R}),$$
from which the result follows immediately.

\section{Large sample properties of the logspline estimator} \label{app:logspline}
The current implementation of our method in the R programming language \citep{R} uses the logspline method by \cite{stone1997polynomial} for marginal density estimation. The asymptotic theory for the logspline estimator is derived by \cite{stone1990large}, but restricted to density functions with compact support. \cite{otneim2016locally} relax this requirement using a truncation argument, so that the requirement of compact support can be replaced by an assumption on the tails of the unknown density not being too heavy.

In particular, \cite{stone1990large} denotes by $\epsilon \in (0,1/2)$ a tuning parameter that determines the asymptotic rate at which new nodes are added to the logspline procedure. If $\epsilon$ is close to zero, new nodes are added quickly to the procedure, and as $\epsilon\rightarrow1/2$, new nodes are added very slowly. \cite{stone1990large} then provides the following asymptotic results (again, under the assumption that the true density $f(\vec{x})$ has compact support):
$$\sqrt{n^{0.5 + \epsilon}}\left(\widehat f_i(x) - f(x)\right) \stackrel{\mathcal{L}}{\rightarrow} N(0, \sigma_1^2),$$
and
$$\sqrt{n^{0.5}}\left(\widehat F_i(x) - F(x)\right) \stackrel{\mathcal{L}}{\rightarrow} N(0, \sigma_2^2).$$
\cite{otneim2016locally} show that these results hold if there exist constants $M>0$, $\gamma > 2\epsilon/(1-2\epsilon)$, and $x_0>0$ such that $f(x)\leq M|x|^{-(5/2+\gamma)}$ for all $|x|>x_0$, so the 'worst case scenario' with respect to assumption \ref{ass:backtrans} when using the logspline estimator for the final back-transformation, is $\epsilon$ being close to zero. In that case, we must require the bandwidths to tend to zero fast enough so that $n^{1/2}h^2\rightarrow 0$, but on the other hand, that will allow $\gamma$ to approach zero, and thus the tail-thickness of the density to approach that of $|x|^{-5/2}$.  

What remains here is to show that these results hold also in the case where the observations are $\alpha$-mixing. This is easily done by replacing the use of the iid central limit theorem (clt) in the proof of Theorem 3 in \cite{stone1990large}, with a corresponding clt that holds under our mixing condition. For example, Theorem A by \cite{peligrad1992central} proves the clt under $\alpha$-mixing provided that the mixing coefficients satisfy $\sum_{n=1}^{\infty}\alpha(n)^{1-2/\delta} < \infty$. This condition follows from our assumption \ref{ass:alpha}.

\section{Supplementary material}
The file \texttt{code.zip}, that accompanies this article, contains the data sets that has been used, as well as routines for implementing the conditional density estimator in the R programming language \citep{R}.

\bibliographystyle{plainnat}
\bibliography{conditionalDensity}   
\end{document}

%% file: settings.tex
\newif\ifcompileiliad 
	\compileiliadfalse 

\newif\ifDownscaledFinalDoc
	\DownscaledFinalDoctrue		

\newif\ifDraft
	\Draftfalse		